\documentclass[11pt]{article}

\usepackage[T1]{fontenc}
\usepackage{microtype}
\usepackage{amsmath,amsthm}
\usepackage{newpxtext,newpxmath}
\usepackage[margin=15pt,font=small,labelfont={bf}]{caption}
\usepackage[margin=1in]{geometry}
\usepackage[english]{babel}
\usepackage[utf8]{inputenc}

\usepackage[backend=biber,style=alphabetic,maxalphanames=3,minalphanames=3,maxbibnames=99,minbibnames=99]{biblatex}
\usepackage{authblk}
\usepackage{graphicx}
\usepackage{algorithm}
\usepackage[noend]{algpseudocode}

\AtBeginDocument{
  }
\usepackage{xcolor}
\usepackage[colorlinks=true,
    citecolor=darkgreen,
    linkcolor=darkred
]{hyperref}
\usepackage{aliascnt}
\definecolor{darkgreen}{rgb}{0,0.5,0.5}
\definecolor{darkred}{rgb}{0.5,0,0}
\newcommand\drop[1]{}
\newcommand*{\newaliastheorem}[3]{
  \newaliascnt{#1}{#2}
  \newtheorem{#1}[#1]{#3}
  \aliascntresetthe{#1}
  \expandafter\newcommand\csname #1autorefname\endcsname{#3}
}
\newtheorem{theorem}{Theorem}
\newaliastheorem{proposition}{theorem}{Proposition}
\newaliastheorem{lemma}{theorem}{Lemma}
\newaliastheorem{corollary}{theorem}{Corollary}
\newaliastheorem{definition}{theorem}{Definition}
\newaliastheorem{conjecture}{theorem}{Conjecture}
\newaliastheorem{example}{theorem}{Example}
\newaliastheorem{remark}{theorem}{Remark}

\addto\extrasenglish{

}

\title{The Balanced Four-Color Theorem}
\author[1,2]{Ken-ichi Kawarabayashi\thanks{Supported by JSPS Kakenhi JP26K21777, JP25K24465, and by JST ASPIRE JPMJAP2302.}}
\author[1]{Hirotaka Yoneda\thanks{Supported by JSPS JP25K24465, JST ASPIRE JPMJAP2302, and JST ACT-X JPMJAX25CT.}}
\author[1]{Masataka Yoneda\thanks{Supported by JSPS JP25K24465 and JST ASPIRE JPMJAP2302.}}
\affil[1]{The University of Tokyo, Tokyo, Japan}
\affil[2]{National Institute of Informatics, Tokyo, Japan}
\date{}

\begin{document}
\maketitle

\pagenumbering{gobble}

\begin{abstract}
    We show that every planar graph with $n \geq 3$ vertices admits a 4-coloring in which each color is used on fewer than $n/2$ vertices. This bound is the best possible. Moreover, such a coloring can be found in $O(n \log n)$ time. We also extend these results to five or more colors and to graphs on general surfaces.
\end{abstract}

\pagenumbering{arabic}
\setcounter{page}{1}

\section{Introduction}

Starting with the celebrated four-color theorem, which states that every planar graph is 4-colorable, the graph coloring problem, which asks for coloring each vertex of a graph $G = (V, E)$ with a minimum number of colors so that no two adjacent vertices have the same color, has been one of the most fundamental problems in graph theory and algorithms. The four-color theorem had been open for more than a century until Appel and Haken \cite{AH77} finally proved it in 1976, and it is now considered one of the most famous theorems in all of mathematics.

In this paper, we revisit the four-color theorem. The proof, as well as its later constructive proofs \cite{RSS+97, IKM+26}, yields only a single coloring, which may be inflexible. Indeed, the only planar graphs that yield a unique 4-coloring are Apollonian networks, characterized by Fowler \cite{Fow98}.

Therefore, we consider a somewhat different direction; here comes into the spotlight: ``How can we obtain a `balanced' 4-coloring?'' This question is related to equitable coloring, a well-studied topic in extremal graph theory since the 1960s \cite{CH63, Erd64, Gru68, HS70, Mey73}, which we briefly discuss below. Moreover, coloring graphs in a balanced way has many applications, including balanced allocation, work sharing, and load balancing, for example, in scheduling \cite{Lei79, Mar04}, clustering \cite{HD78}, register allocation \cite{Cha82}, and bandwidth allocation \cite{Pas23}.

\subsection{Equitable Coloring}
For a graph $G = (V, E)$, its $k$-coloring is \emph{equitable} if the number of vertices of each color differs by at most one. One of the most fundamental results in extremal graph theory concerns equitable coloring: every graph with maximum degree $\Delta$ has an equitable $k$-coloring for all $k \geq \Delta+1$, which was conjectured by Erd\H{o}s in 1964 \cite{Erd64} and proved by Hajnal and Szemer\'{e}di in 1970 \cite{HS70, KK08} (indeed, this result was originally proved in terms of clique packing, i.e., a spanning collection of vertex-disjoint cliques). There have also been a number of studies on equitable coloring beyond graphs with maximum degree $\Delta$; one example is the work in \cite{KNP05} on $d$-degenerate graphs.

\subsection{The Four-Color Theorem and ``Balanced'' Coloring}

Returning to the original motivation, we ask whether a balanced $4$-coloring of a planar graph exists. As an example, consider the bipartite graph $K_{2,n-2}$ with one extra edge, say $\{v_1, v_2\}$, in the partite set of two vertices. This graph has no equitable $4$-coloring (when $n \geq 7$), because if we use color $1$ on $v_1$ and color $2$ on $v_2$, we can only use colors $3$ and $4$ on the remaining $n-2$ vertices. Thus, on one hand, the answer to the question is ``no''. However, on the other hand, if we interpret it as the question of whether each color class can be meaningfully small, the answer might be ``yes'' --- this graph still has a $4$-coloring in which colors $3, 4$ are each used on fewer than $n/2$ vertices (\autoref{fig:lowerbound}); every planar graph might admit a meaningfully balanced $4$-coloring.

\begin{figure}[htbp]
    \centering
    \includegraphics[width=\linewidth]{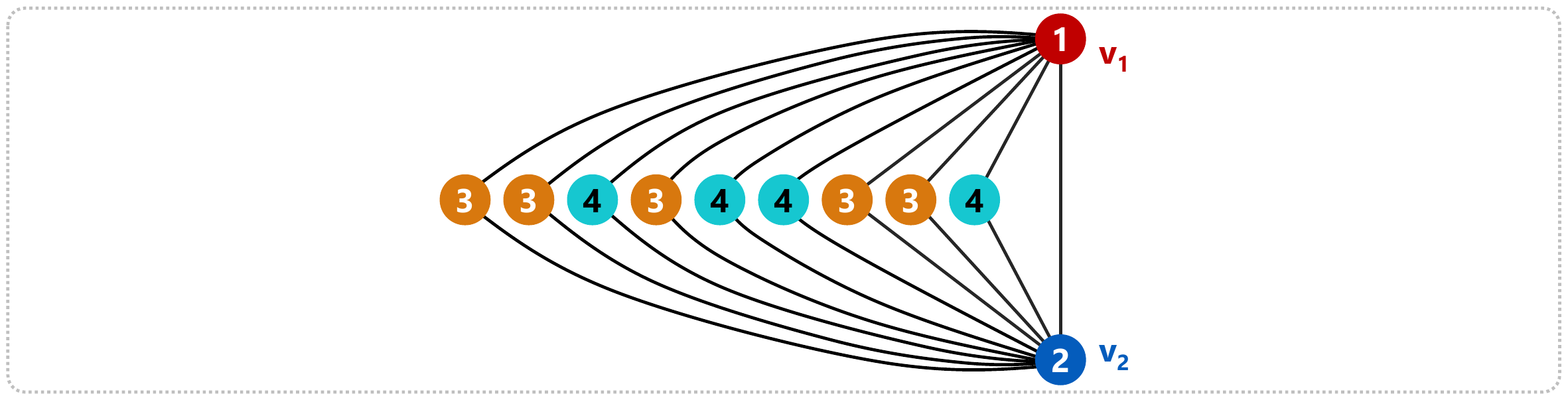}
    \caption{A $4$-coloring of the graph $K_{1,1,n-2}$ with $n = 11$. The number written in each vertex is the assigned color. At least $5$ vertices receive the same color.}
    \label{fig:lowerbound}
\end{figure}

Despite the motivation, to the best of our knowledge, there was no prior research on how balanced $4$-colorings of planar graphs can be achieved. One reason may be that the complexity of the proof of the four-color theorem impedes attempts to strengthen the theorem; \cite{IKM+26} already shows that a generalization of the four-color theorem is highly non-trivial. Another reason may be that the graph theory community has mainly focused on equitable coloring: for example, \emph{bounded coloring} \cite{HHK93}, which aims to minimize the number of colors when the size of each color class is bounded, has not been studied as extensively as equitable coloring.

Let us point out that, given the relationship between equitable coloring and fair division \cite{HH22}, a hot topic in algorithmic game theory, achieving approximately balanced colorings is a natural line of research in this sense, just as fair division focuses on approximate fairness \cite{AAB+23}.

\subsection{Color-Critical Graphs}

It is essential to mention the concept of \emph{color-critical graphs} to discuss the background of this question. A graph $G$ is called $(k+1)$-critical if $G$ is not $k$-colorable but any proper subgraph is. Using this concept, we can create a $k$-colorable graph (which is less obvious than $K_{1,1,n-2}$) such that every $k$-coloring has a color class of size at least around $\frac{n}{k-2}$; indeed, we use a $(k+1)$-critical graph $H$ as a ``gadget''. From $H$, we delete an arbitrary edge $e = \{a, b\}$, and instead we add a new vertex $c$ and draw an edge $\{a, c\}$. Then, we merge this graph with $K_{2,t}$, gluing along vertices $\{b, c\}$ and the partite set of two vertices; see \autoref{fig:grotzsch} for an example.

Since every $k$-coloring of $H-e$ must assign the same color to $a$ and $b$ (due to criticality), it must assign different colors to $b$ and $c$. Therefore, only the remaining $k-2$ colors can be used on the partite set of $t$ vertices, and hence every $k$-coloring has a color class of size at least $\frac{t}{k-2}$.

On the other hand, color-critical graphs lead to deep (algorithmic) results in topological graph theory, especially for graphs on surfaces. Let us contrast the plane and a general surface:

\begin{itemize}
    \item Planar graphs: by the four-color theorem, every graph is $4$-colorable (so no 5-color-critical planar graphs exist).
    \item Graphs on a surface: although the chromatic number can be very large, by the deep result of Thomassen \cite{Tho97}, there are only finitely many $6$-color-critical graphs that can be embedded in any orientable surface. Therefore, one can check if a graph embedded in a fixed surface is $5$-colorable in polynomial time.
\end{itemize}

Therefore, exploring $5$-coloring for graphs on a surface is a significant question, as is $4$-coloring for planar graphs. We shall present balanced colorings for each of these cases.

\begin{figure}
    \centering
    \includegraphics[width=\linewidth]{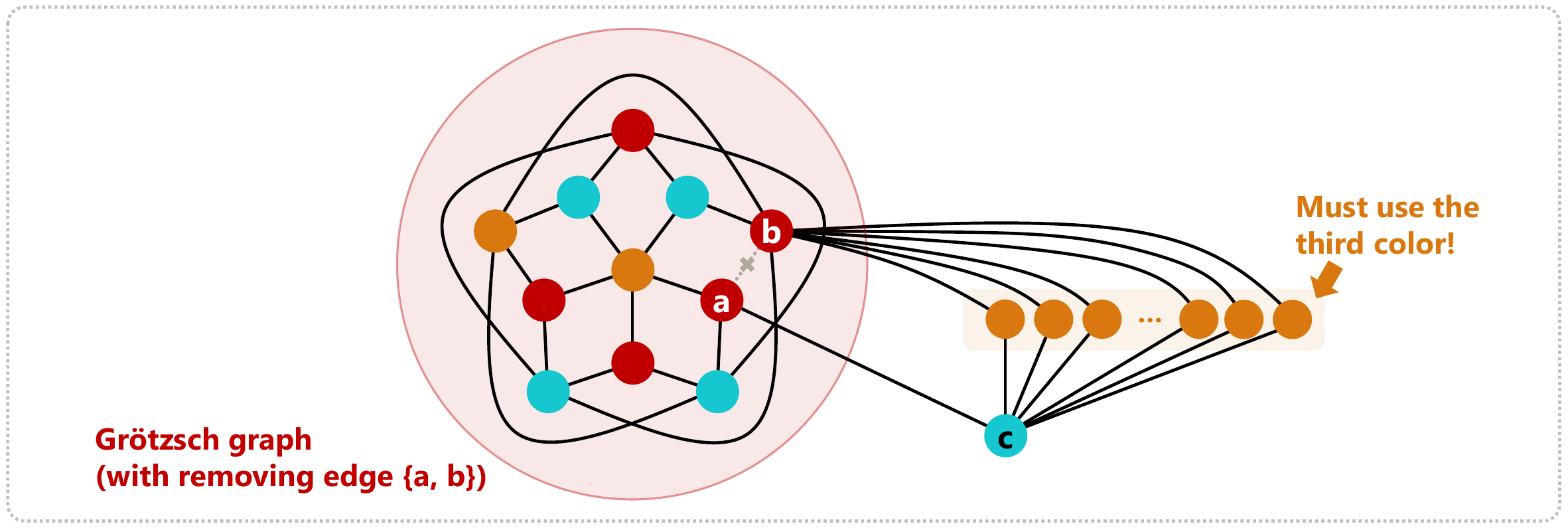}
    \caption{An example in which every 3-coloring is unbalanced. Constructed by gluing $K_{2,t}$ and a $4$-critical graph (the Gr\"{o}tzsch graph) minus one edge. This can be embedded in a projective plane.}
    \label{fig:grotzsch}
\end{figure}

\subsection{Our Contributions}

In this paper, we address the question of how balanced the $4$-coloring of planar graphs can be. Our main result is described as follows:

\begin{description}
    \item[\textbf{Main Result.}] Every planar graph with $n \geq 3$ vertices admits a $4$-coloring such that each color is used on fewer than $n/2$ vertices, and such a coloring can be computed in $O(n \log n)$ time.
\end{description}

We will not dig into the details of the proof of the four-color theorem; rather, we use this theorem as a ``black box'' tool. Our strategy is as follows: we first compute an arbitrary $4$-coloring with an $O(n \log n)$-time algorithm by \cite{IKM+26}, and then we iteratively improve it to make it more balanced by utilizing \emph{Kempe chains} \cite{Kem79}, which is a tool well-known for the four-color theorem. We refer to it as the \emph{Kempe change}, the operation of interchanging two colors in a Kempe chain; this idea has also been studied in reconfiguration problems for graph coloring beyond planar graphs \cite{Moh06, BHI+20}. In short, we show that $O(\log n)$ iterations of Kempe changes is enough to make the coloring sufficiently balanced; since each iteration takes $O(n)$ time, we achieve $O(n \log n)$ time in total.

Strictly speaking, we show the following technical theorem (for $k$-coloring) involving a new parameter $\alpha$, which allows us to extend the result to other surfaces; we use $\alpha = -4$ for planar graphs. With $k = 4$, by combining this theorem with \cite{IKM+26}, we obtain our main result above.

\begin{theorem} \label{thm:main}
    Let $\alpha \geq -4$ be a constant parameter and $k \geq 3$. Let $G$ be a $k$-colorable graph of $n \geq 3$ vertices, such that for all its bipartite subgraphs $H$ with 3 or more vertices, $|E(H)| \leq 2 |V(H)| + \alpha$ holds. Given $G = (V, E)$ together with a $k$-coloring $\varphi: V \to \{1, \dots, k\}$, one can compute a $k$-coloring of $G$ such that each color is used on at most $\left\lceil \frac{n+\alpha+2}{k-2} \right\rceil$ vertices, in $O(kn \log n)$ time.
\end{theorem}

We note that the threshold is $\left\lceil \frac{n-2}{k-2} \right\rceil$ for planar graphs, due to $\alpha = -4$. Meanwhile, for a planar graph $K_{1,1,n-2}$ (see \autoref{fig:lowerbound}), one color is used on $\left\lceil \frac{n-2}{k-2} \right\rceil$ or more vertices in every $k$-coloring. Therefore, the bound above is the best possible, for all $n \geq 3$ and $k \geq 3$.

\subsection{Organization of the Paper}

In \autoref{sec:prelim}, we present the preliminaries needed for the proof. In \autoref{sec:overview}, we present an overview of the proof and the algorithm. In \autoref{sec:proof}, we present the proof of the main technical theorem and the algorithm. In \autoref{sec:consequences}, we state the implications of the main technical theorem for planar graphs and other surfaces (e.g., projective plane and torus). Finally, in \autoref{sec:conclusion}, we conclude this paper by mentioning several conjectures.

\section{Preliminaries} \label{sec:prelim}

\paragraph{\textbf{Convention and Notation.}}
In this paper, we consider an undirected simple graph $G = (V, E)$, which has no self-loops and no multiple edges. The sets of vertices and edges of $G$ are also denoted by $V(G)$ and $E(G)$, respectively. The degree of a vertex $v \in V$ in $G$ is denoted by $d(v)$ or $d_G(v)$. The \emph{induced subgraph} of $G$ on a vertex set $S \subseteq V$ is denoted by $G[S]$. We identify each \emph{component} of $G$ with its vertex set $C \subseteq V$.

\paragraph{\textbf{Kempe Change.}}
Let $\varphi: V \to \{1, \dots, k\}$ be a coloring of $G = (V, E)$. Let $V_i := \{v \in V : \varphi(v) = i\}$ be the set of vertices of color $i$. A \emph{Kempe change} is an operation that chooses two colors $a, b$ and a component $C \subseteq V_a \cup V_b$ of the induced subgraph $G[V_a \cup V_b]$, and then recolors each vertex in $C$ from $a$ to $b$, and from $b$ to $a$. \autoref{fig:flip-operation} shows an example. Trivially, the resulting coloring after Kempe change is proper, i.e., no two adjacent vertices have the same color.

\begin{figure}[htbp]
    \centering
    \includegraphics[width=\linewidth]{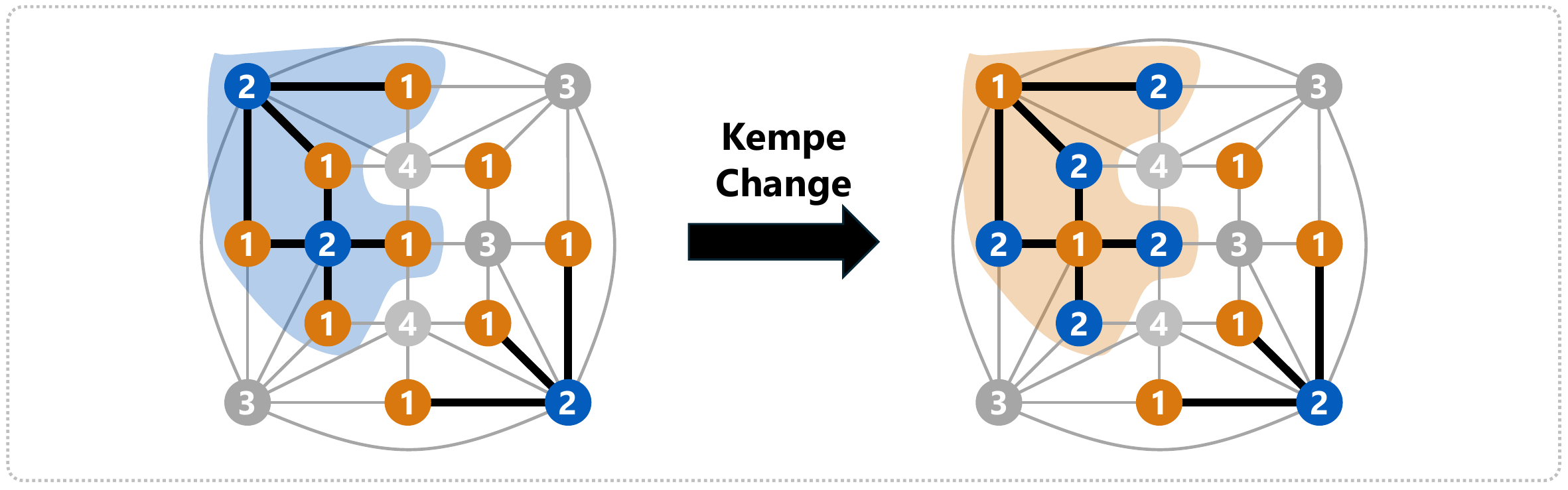}
    \caption{An example of a Kempe change, which is applied to a component of colors $1$ and $2$.}
    \label{fig:flip-operation}
\end{figure}

\section{Strategies and Intuitions} \label{sec:overview}

In this section, we sketch the idea of our proof for the main result: every planar graph $G = (V, E)$ with $n \geq 3$ vertices admits a $4$-coloring such that each color is used on fewer than $n/2$ vertices.

\subsection{Overall Strategy}

The overall strategy is to first obtain an arbitrary $4$-coloring $\varphi: V \to \{1, 2, 3, 4\}$, and then repeatedly improve the coloring $\varphi$ by Kempe changes. For simplicity, we always assume, w.l.o.g., that color $1$ is the most frequently used. Also, for $i = 1, 2, 3, 4$, let $V_i := \{v \in V : \varphi(v) = i\}$ be the set of vertices of color $i$. In order to make sure that $\varphi$ finally becomes a coloring in which each color is used on fewer than $n/2$ vertices, all we need to prove is the following: ``if $|V_1| \geq n/2$, there exists a Kempe change involving color $1$ that decreases the size of the largest color class.''

\subsection{Idea 1: ``All Connected'' is Impossible} \label{subsec:overview-2}

When it is impossible to improve the coloring by a Kempe change, what would the coloring look like? For $i = 2, 3, 4$, let $H_i := G[V_1 \cup V_i]$ be the induced subgraph on vertices of colors $1$ and $i$. A leading ``impossible'' example is when $H_2, H_3, H_4$ are all connected --- in this case, a Kempe change on $H_i$ simply interchanges colors $1$ and $i$ in the entire coloring $\varphi$, resulting in no improvement.

The good news is that, as long as $|V_1| \geq n/2$, this case cannot happen. We prove this in what follows, but before that, we define a graph $H := H_2 \cup H_3 \cup H_4$. Note that $H$ consists of all vertices in $G$, and the edges in $G$ where one endpoint has color $1$. Then, $H$ is a bipartite graph, which can be divided into two parts $V_1$ and $V_2 \cup V_3 \cup V_4$.

\begin{proposition}[\cite{MT01}] \label{prop:bipartite-max}
    Every bipartite planar graph with $n \geq 3$ vertices has at most $2n-4$ edges.
\end{proposition}

\begin{proposition} \label{prop:not-all-connected}
    Assuming $|V_1| \geq n/2$, at least one of $H_2, H_3, H_4$ is not connected.
\end{proposition}

\begin{proof}
    Suppose $H_2, H_3, H_4$ are all connected. Then we have $|E(H_i)| \geq |V(H_i)| - 1 = |V_1| + |V_i| - 1$ for each $i$. Therefore, the total number of edges is at least
    \begin{equation*}
        |E(H)| \geq \sum_{i=2}^4 (|V_1|+|V_i|-1) = 2 |V_1| + \sum_{i=1}^4 |V_i| - 3 = 2 |V_1| + n - 3
    \end{equation*}
    However, $|E(H)| \leq 2n-4$ due to \autoref{prop:bipartite-max}. Hence, we obtain $2 |V_1| + n - 3 \leq 2n-4$, which yields $|V_1| \leq \frac{n-1}{2}$. Since $|V_1|$ and $n$ are integers, the actual bound is $|V_1| < n/2$, which contradicts the assumption (see \autoref{fig:all-connected}). Therefore, at least one of $H_2, H_3, H_4$ is not connected.
\end{proof}

\begin{figure}[b]
    \centering
    \includegraphics[width=\linewidth]{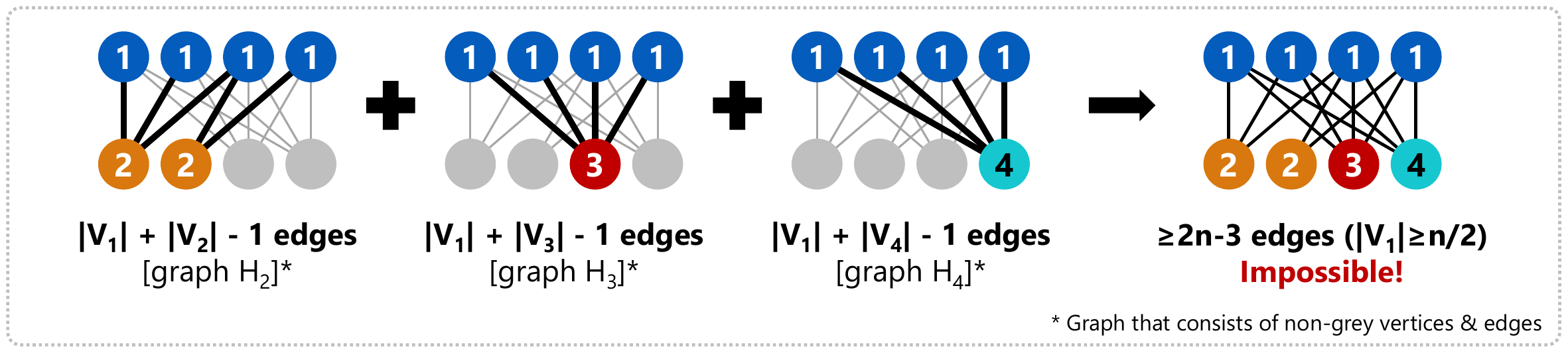}
    \caption{The sketch of the proof, along with the definition of the graph $H_i$. Note that this figure does not include the edges between colors $2, 3,$ and $4$.}
    \label{fig:all-connected}
\end{figure}

\subsection{Idea 2: The Penalty for Deviations}

While \autoref{prop:not-all-connected} gives the intuition that the $n/2$ bound may be achievable, there might be some cases where, even if some $H_i$ is not connected, no Kempe changes can improve the coloring. However, in this case, the number of vertices of color $1$ versus color $i$ must be highly imbalanced across the components of $H_i$. In particular, one can show that there must exist a component $C \subseteq V_1 \cup V_i$ of $H_i$ such that $|C \cap V_i| \geq |C \cap V_1|$ (cf. \autoref{lem:penalty2}).

Let us consider what the existence of such a component $C$ triggers. As seen in the proof of \autoref{prop:not-all-connected}, the bound on improving the coloring by Kempe change is based on the fact that $H$ can have at most $2n-4$ edges. Apparently, the more edges in $H$, the less freedom to improve the coloring by Kempe changes. However, the following happens:

\begin{itemize}
    \item If $H_i[C]$ is not a tree: The set of possible Kempe changes remains unchanged even if we remove all the edges not in a certain spanning tree of $H_i[C]$. Therefore, the situation is substantially the same as when $H$ has at most $2n-5$ edges.
    \item If $H_i[C]$ is a tree: Since there are $|C|-1$ edges in the tree and the condition $|C \cap V_i| \geq |C \cap V_1|$ holds, there must exist a vertex $v \in C$ of color $i$ which has degree at most $1$. Then, $v$ has degree at most $1$ also in $H$. However, for a bipartite planar graph $H$ that has a vertex of degree at most $1$, it can have at most $2n-5$ edges when $n \geq 4$.
\end{itemize}

Therefore, both cases turn out to be disadvantageous for creating a situation where no Kempe change can improve the coloring, and intuitively, for this reason, creating such a situation is impossible when $|V_1| \geq n/2$ (we prove this in \autoref{sec:proof} in a way that analyzes ``excessive edges'' and ``low-degree vertices''). As a result, by repeatedly applying Kempe changes while $|V_1| \geq n/2$, we can obtain a $4$-coloring such that each color is used on fewer than $n/2$ vertices.

\section{The Proof and Algorithm}\label{sec:proof}

Now, we start proving the main technical theorem for a general number of colors $k$ and a generalized graph class parameterized by $\alpha \geq -4$. Note that we use $\alpha = -4$ for planar graphs, by \autoref{prop:bipartite-max}.

\begingroup
\def\thetheorem{\ref{thm:main}}
\begin{theorem}[restated]
    Let $\alpha \geq -4$ be a constant parameter and $k \geq 3$. Let $G$ be a $k$-colorable graph of $n \geq 3$ vertices, such that for all its bipartite subgraphs $H$ with 3 or more vertices, $|E(H)| \leq 2 |V(H)| + \alpha$ holds. Given $G = (V, E)$ together with a $k$-coloring $\varphi: V \to \{1, \dots, k\}$, one can compute a $k$-coloring of $G$ such that each color is used on at most $\left\lceil \frac{n+\alpha+2}{k-2} \right\rceil$ vertices, in $O(kn \log n)$ time.
\end{theorem}
\addtocounter{theorem}{-1}
\endgroup

\subsection{Part 1: Analysis of Components} \label{sec:proof-1}

In this section, before proceeding to the explanation of our algorithm, we formulate the conditions of connected components in the graph stated in \autoref{thm:main}.

\paragraph{The Setup.} For the given $k$-coloring $\varphi : V \to \{1, \dots, k\}$, let $V_i := \{v \in V : \varphi(v) = i\}$ for $i = 1, \dots, k$. If $\varphi$ does not yet meet the output condition, we can assume w.l.o.g. that $|V_1| > \left\lceil \frac{n+\alpha+2}{k-2} \right\rceil$.

As described in \autoref{sec:overview}, we define the graph $H_i := G[V_1 \cup V_i]$ for $i = 2, \dots, k$ and $H := H_2 \cup \dots \cup H_k$. We analyze the structure of each $H_i$. For each $i = 2, \dots, k$, let $C_{i,1}, \dots, C_{i,q_i} \subseteq V(H_i)$ be the components of $H_i$, where $q_i$ is the number of components. For each component $C_{i,j}$, we define the values $a_{i,j}, b_{i,j}, e_{i,j}, p_{i,j}$ as in \autoref{tab:variable}.

\begin{table}[t]
    \small
    \centering
    \caption{The list of variables that we define for each component $C_{i,j}$.}
    \begin{tabular}{l|l|p{9.5cm}}
    \hline
        Value & Definition & Meaning \\ \hline
        $a_{i, j}$ & $= |C_{i,j} \cap V_1|$ & Number of vertices with color $1$ \\
        $b_{i, j}$ & $= |C_{i,j} \cap V_i|$ & Number of vertices with color $i$ \\
        $e_{i, j}$ & $= |E(G[C_{i,j}])|$ & Number of edges \\
        $p_{i, j}$ & $= \sum_{v \in C_{i,j} \cap V_i} \max(2 - d_{H_i}(v), 0)$ & The ``penalty'' for low-degree color $i$ vertices \\ \hline
    \end{tabular}
    \label{tab:variable}
\end{table}

The key factor in deciding whether we can improve the coloring $\varphi$ lies in $a_{i,j} - b_{i,j}$. This is because, if we apply Kempe change in the component $C_{i,j}$, the number of vertices of colors $1$ and $i$ decreases and increases by $a_{i,j} - b_{i,j}$, respectively. Therefore, the goal of this section is to derive some useful conditions on the distribution of $a_{i,j} - b_{i,j}$.

\begin{lemma} \label{lem:bipartite-max-refine}
    $|E(H)| \leq (2n+\alpha) - \sum_{v \in V_2 \cup \dots \cup V_k} \max(2 - d_H(v), 0)$.
\end{lemma}

\begin{proof}
    First, by $\alpha \geq -4$ and $n \geq 3$, we have $|V_1| > \left\lceil \frac{n+\alpha+2}{k-2} \right\rceil \geq \left\lceil \frac{1}{k-2} \right\rceil \geq 1$. Thus, $|V_1| \geq 2$ holds. Next, let $S = \{v \in V_2 \cup \dots \cup V_k : d_H(v) < 2\}$. Then, since $H$ is a bipartite graph that can be divided into two parts $V_1$ and $V_2 \cup \dots \cup V_k$, we have $|E(H)| = |E(H[V \setminus S])| + \sum_{v \in S} d_H(v)$. Here:
    \begin{itemize}
        \item $|E(H[V \setminus S])| \leq 2 |V \setminus S| + \alpha$ holds. This is because, if $|V \setminus S| \geq 3$ then the inequality holds by assumption, and otherwise, since $|V_1| \geq 2$, the only possibility is that $V \setminus S = V_1$ and $|V_1| = 2$, and in this case, $|E(H[V \setminus S])| = 0 \leq 2 |V \setminus S| + \alpha$ due to $\alpha \geq -4$.
        \item Also, we have $\sum_{v \in S} d_H(v) = \sum_{v \in S} \{2 - (2 - d_H(v))\} = 2 |S| - \sum_{v \in S} (2 - d_H(v))$.
    \end{itemize}
    Summing up these two results, we have $|E(H)| \leq (2 |V \setminus S| + \alpha) + \{2 |S| - \sum_{v \in S} (2 - d_H(v))\} = 2n + \alpha - \sum_{v \in V_2 \cup \dots \cup V_k} \max(2 - d_H(v), 0)$.
\end{proof}

\begin{lemma} \label{lem:exy2}
    $\sum_{i=2}^k \sum_{j=1}^{q_i} (e_{i,j} + p_{i,j}) \leq 2n+\alpha$.
\end{lemma}

\begin{proof}
    We rewrite the result of \autoref{lem:bipartite-max-refine} using $e_{i,j}$ and $p_{i,j}$. First, $|E(H)| = \sum_{i=2}^k \sum_{j=1}^{q_i} e_{i,j}$ holds by definition of $H$. Also, for $i \in \{2, \dots, k\}$, observe that $d_H(v) = d_{H_i}(v)$ holds for each $v \in V_i$. Then,
    \begin{equation*}
        \sum_{v \in V_2 \cup \dots \cup V_k} \max(2 - d_H(v), 0) = \sum_{i=2}^k \sum_{j=1}^{q_i} \sum_{v \in C_{i,j} \cap V_i} \max(2 - d_{H_i}(v), 0) = \sum_{i=2}^k \sum_{j=1}^{q_i} p_{i,j}
    \end{equation*}
    and therefore, from \autoref{lem:bipartite-max-refine}, we obtain $\sum_{i=2}^k \sum_{j=1}^{q_i} e_{i,j} \leq 2n+\alpha - \sum_{i=2}^k \sum_{j=1}^{q_i} p_{i,j}$.
\end{proof}

\begin{lemma} \label{lem:exy3}
    $e_{i,j} + p_{i,j} \geq \max(a_{i,j}+b_{i,j}-1, 2b_{i,j})$.
\end{lemma}

\begin{proof}
    It suffices to prove that $e_{i,j}+p_{i,j} \geq a_{i,j}+b_{i,j}-1$ and $e_{i,j}+p_{i,j} \geq 2b_{i,j}$.
    \begin{itemize}
        \item Since $H_i[C_{i,j}]$ is connected, $e_{i,j} \geq |C_{i,j}|-1 = a_{i,j} + b_{i,j} - 1$. As $p_{i,j} \geq 0$ is obvious, we have $e_{i,j} + p_{i,j} \geq a_{i,j} + b_{i,j} - 1$.
        \item Also, $e_{i,j} + p_{i,j} = \sum_{v \in C_{i,j} \cap V_i} d_{H_i}(v) + \sum_{v \in C_{i,j} \cap V_i} \max(2 - d_{H_i}(v), 0) = \sum_{v \in C_{i,j} \cap V_i} \max(d_{H_i}(v), 2) \geq 2 |C_{i,j} \cap V_i| = 2b_{i,j}$. \qedhere
    \end{itemize}
\end{proof}

\begin{lemma} \label{lem:sum-ab}
    $\sum_{i=2}^k \sum_{j=1}^{q_i} (a_{i,j}+b_{i,j}) = (k-2)|V_1|+n$.
\end{lemma}

\begin{proof}
    By definition, $\sum_{j=1}^{q_i} a_{i,j} = |V_1|$ and $\sum_{j=1}^{q_i} b_{i,j} = |V_i|$. Therefore, 
    
    \begin{align*}
    \sum_{i=2}^k \sum_{j=1}^{q_i} (a_{i,j}+b_{i,j}) = \sum_{i=2}^k (|V_1| + |V_i|) &= (k-2) |V_1| + (|V_1| + \dots + |V_k|) \\
    &= (k-2) |V_1| + n.
    \qedhere
    \end{align*}
\end{proof}

\begin{lemma} \label{lem:penalty1}
    $\sum_{i=2}^k \sum_{j=1}^{q_i} \max(b_{i,j}-a_{i,j}, -1) \leq -(k-2) |V_1| + (n+\alpha)$.
\end{lemma}

\begin{proof}
    We obtain $\max(b_{i,j}-a_{i,j}, -1) \leq (e_{i,j}+p_{i,j}) - (a_{i,j}+b_{i,j})$ by subtracting $a_{i,j}+b_{i,j}$ from both sides of the inequality in \autoref{lem:exy3}. Therefore:
    \begin{eqnarray*}
        \sum_{i=2}^k \sum_{j=1}^{q_i} \max(b_{i,j}-a_{i,j}, -1) &\leq& \sum_{i=2}^k \sum_{j=1}^{q_i} (e_{i,j}+p_{i,j}) - \sum_{i=2}^k \sum_{j=1}^{q_i} (a_{i,j}+b_{i,j}) \\ &\leq& (2n+\alpha) - \{(k-2)|V_1|+n\}
    \end{eqnarray*}
    where the last inequality is due to \autoref{lem:exy2} and \autoref{lem:sum-ab}. Note that $(2n+\alpha) - \{(k-2)|V_1|+n\} = -(k-2) |V_1| + (n+\alpha)$.
\end{proof}

The implication of \autoref{lem:penalty1} is that the larger $|V_1|$ is, the more components $C_{i,j}$ such that $a_{i,j} > b_{i,j}$ there are. Especially, given that $|V_1| > \left\lceil \frac{n+\alpha+2}{k-2} \right\rceil$ (which implies $|V_1| \geq \frac{n+\alpha+2}{k-2} + 1$), \autoref{lem:penalty1} states that $\sum_{i=2}^k \sum_{j=1}^{q_i} \max(b_{i,j}-a_{i,j}, -1) \leq -k$. This means that $\sum_{j=1}^{q_i} \max(b_{i,j}-a_{i,j}, -1) \leq -2$ for some $i$, and hence $H_i$ has at least two components.

\subsection{Part 2: The Algorithm} \label{sec:proof-2}

In this section, we derive a fast algorithm to obtain a $k$-coloring that is balanced enough. The strategy is to repeat the steps that make the coloring $\varphi$ more balanced, using Kempe changes.

For a subset $S \subseteq \{1, \dots, q_i\}$, if we apply Kempe changes to all $C_{i,j}$ such that $j \in S$, then the number of vertices of color $1$ becomes $|V_1| - \sum_{j \in S} (a_{i,j}-b_{i,j})$, and the number of vertices of color $i$ becomes $|V_i| + \sum_{j \in S} (a_{i,j}-b_{i,j})$. From now on, we discuss whether it is possible to choose a subset $S \subseteq \{1, \dots, q_i\}$ on which to apply Kempe changes, so that the coloring $\varphi$ improves.

We define $\delta := |V_1| - \left(\frac{n+\alpha+1}{k-2}+1\right)$ to express how excessive the number of vertices of color $1$ is compared to the threshold; e.g., $\delta = \frac{1}{k-2}$ when $|V_1| = \frac{n+\alpha+2}{k-2}+1$ (the case with minimum $|V_1|$).

\begin{lemma} \label{lem:penalty2}
    There exists some $i$ such that $\sum_{j=1}^{q_i} \max(b_{i,j}-a_{i,j}, -1) \leq -\left(\frac{1}{2} \delta + 1\right)$.
\end{lemma}

\begin{proof}
    By \autoref{lem:penalty1}, there exists some $i$ such that
    
    $$\sum_{j=1}^{q_i} \max(b_{i,j}-a_{i,j}, -1) \leq \frac{-(k-2)|V_1|+(n+\alpha)}{k-1} = -\left(\frac{k-2}{k-1} \delta +1\right) \leq -\left(\frac{1}{2} \delta + 1\right)$$
    
    \noindent
    where the final inequality is due to $k \geq 3$.
\end{proof}

\begin{lemma} \label{lem:max-cond}
    For $i \in \{2, \dots, k\}$ and $x > 0$, if $\sum_{j=1}^{q_i} \max(b_{i,j}-a_{i,j}, -1) \leq -x$, then the following holds:
    \begin{enumerate}
        \item[(a)] $|V_1| > |V_i|$.
        \item[(b)] $a_{i,j} - b_{i,j} \leq (|V_1|-|V_i|)-(x-1)$ for all $j \in \{1, \dots, q_i\}$.
    \end{enumerate}
\end{lemma}

\begin{proof}
    First, $|V_1| = \sum_{j=1}^{q_i} a_{i,j}$ and $|V_i| = \sum_{j=1}^{q_i} b_{i,j}$ by definition; (a) is shown by:
    \begin{equation*}
        |V_i| - |V_1| = \sum_{j=1}^{q_i} (b_{i,j}-a_{i,j}) \leq \sum_{j=1}^{q_i} \max(b_{i,j}-a_{i,j}, -1) \leq -x < 0
    \end{equation*}
    It remains to prove (b). By assumption, we have:
    \begin{equation*}
        \sum_{j=1}^{q_i} \{\max(b_{i,j}-a_{i,j}, -1) - (b_{i,j}-a_{i,j})\} \leq -x - (|V_i|-|V_1|) = (|V_1|-|V_i|) - x
    \end{equation*}
    Since $\max(b_{i,j}-a_{i,j}, -1) - (b_{i,j}-a_{i,j})$ is non-negative for all $j$, $\max(b_{i,j}-a_{i,j}, -1) - (b_{i,j}-a_{i,j}) \leq (|V_1| - |V_i|) - x$ must hold for all $j$. Therefore, $(-1) - (b_{i,j}-a_{i,j}) \leq (|V_1| - |V_i|) - x$ also holds, which means that $a_{i,j}-b_{i,j} \leq (|V_1|-|V_i|)-(x-1)$. This proves (b).
\end{proof}

\begin{lemma} \label{lem:subsetsum1}
    Let $s, t$ be real numbers such that $s > t \geq 0$. Let $d_1, \dots, d_m$ be (not necessarily non-negative) real numbers. If $d_1 + \dots + d_m = s$ and $d_1, \dots, d_m \leq t$, then there exists a subset $S \subseteq \{1, \dots, m\}$ such that $\frac{s-t}{2} \leq \sum_{i \in S} d_i \leq \frac{s+t}{2}$.
\end{lemma}

\begin{proof}
    Let $f(i) = d_1 + \dots + d_i \ (i \in \{0, \dots, m\})$. Then, there exists $i \in \{1, \dots, m\}$ such that $f(i-1) < \frac{s-t}{2}$ and $f(i) \geq \frac{s-t}{2}$. This is because $f(0) = 0 < \frac{s-t}{2}$ and $f(m) = s \geq \frac{s-t}{2}$. For such $i$, we have $f(i) = f(i-1)+d_i < \frac{s-t}{2} + t = \frac{s+t}{2}$. Therefore, by choosing $S = \{1, \dots, i\}$, we have $\frac{s-t}{2} \leq \sum_{i \in S} d_i \leq \frac{s+t}{2}$.
\end{proof}

\begin{lemma} \label{lem:subsetsum2}
    For $i \in \{2, \dots, k\}$ and $x > 1$, if $\sum_{j=1}^{q_i} \max(b_{i,j}-a_{i,j}, -1) \leq -x$, then there exists a subset $S \subseteq \{1, \dots, q_i\}$ such that $\frac{x-1}{2} \leq \sum_{j \in S} (a_{i,j}-b_{i,j}) \leq (|V_1|-|V_i|) - \frac{x-1}{2}$.
\end{lemma}

\begin{proof}
    Because of \autoref{lem:max-cond}, we can apply \autoref{lem:subsetsum1} with $m := q_i, d_j := a_{i,j}-b_{i,j}$, and $s = |V_1|-|V_i|, t = (|V_1|-|V_i|)-(x-1)$. Note that $t \geq 0$ holds because otherwise $d_1, \dots, d_m \leq t < 0$ and $d_1 + \dots + d_m = s > 0$ would contradict. The consequence is that there exists a subset $S \subseteq \{1, \dots, q_i\}$ such that $\frac{s-t}{2} \leq \sum_{j \in S} (a_{i,j}-b_{i,j}) \leq \frac{s+t}{2}$. Here, $\frac{s-t}{2} = \frac{x-1}{2}$ and $\frac{s+t}{2} = (|V_1|-|V_i|) - \frac{x-1}{2}$.
\end{proof}

\begin{lemma} \label{lem:subsetsum3}
    For some $i$, there exists a subset $S \subseteq \{1, \dots, q_i\}$ such that $\frac{1}{4} \delta \leq \sum_{j \in S} (a_{i,j}-b_{i,j}) \leq (|V_1|-|V_i|) - \frac{1}{4} \delta$.
\end{lemma}

\begin{proof}
    By \autoref{lem:penalty2}, there exists some $i$ such that $\sum_{j=1}^{q_i} \max(b_{i,j} - a_{i,j}, -1) \leq -\left(\frac{1}{2} \delta + 1\right)$. By assigning $x = \frac{1}{2} \delta + 1$ to \autoref{lem:subsetsum2}, we obtain the following: there exists a subset $S \subseteq \{1, \dots, q_i\}$ such that
    \begin{equation*}
        \frac{1}{4}\delta \leq \sum_{j\in S}(a_{i,j}-b_{i,j}) \leq (|V_1|-|V_i|)-\frac{1}{4}\delta \qedhere
    \end{equation*}
\end{proof}

The implication of \autoref{lem:subsetsum3} is that the more unbalanced the coloring $\varphi$ is, the more improvement we can make in one batch of Kempe changes. Based on these lemmas, we can construct an algorithm to find a sufficiently balanced coloring. The algorithm is formally given in \autoref{alg:main}.

\begin{algorithm}[htbp]
    \caption{The algorithm which, given a graph $G = (V, E)$ and its $k$-coloring $\varphi: V \to \{1, \dots, k\}$, computes a $k$-coloring that each color is used on at most $\left\lceil \frac{n+\alpha+2}{k-2} \right\rceil$ vertices.}
    \label{alg:main}
    \begin{algorithmic}[1]
        \While {some color is used on more than $\left\lceil \frac{n+\alpha+2}{k-2} \right\rceil$ vertices}
            \State rearrange colors in $\varphi$ so that color $1$ becomes the most frequently used color
            \State compute $V_i, H_i, C_{i,j}, a_{i,j}, b_{i,j}$ and $\delta := |V_1| - \left(\frac{n+\alpha+1}{k-2}+1\right)$
            \State choose $i \in \{2, \dots, k\}$ such that $\sum_{j=1}^{q_i} \max(b_{i,j}-a_{i,j}, -1) \leq -\left(\frac{1}{2} \delta + 1\right)$ (exists by \autoref{lem:penalty2})
            \State choose $S \subseteq \{1, \dots, q_i\}$ such that $\frac{1}{4} \delta \leq \sum_{j \in S} (a_{i,j}-b_{i,j}) \leq (|V_1|-|V_i|) - \frac{1}{4} \delta$ (exists by \autoref{lem:subsetsum3})
            \State apply Kempe change for component $C_{i,j}$, for each $j \in S$
        \EndWhile
        \State \Return $\varphi$
    \end{algorithmic}
\end{algorithm}

\subsection{Part 3: Analysis of Time Complexity} \label{sec:proof-3}

\begin{lemma} \label{lem:num-loops}
    In \autoref{alg:main}, at most $k \cdot \lceil \log_{4/3} (nk) \rceil$ loops will occur.
\end{lemma}

\begin{proof}
    We observe that, after processing line 6 of \autoref{alg:main}, $|V_1|$ and $|V_i|$ both become at most $|V_1| - \frac{1}{4} \delta$. Therefore, $\delta$ is non-increasing throughout all loops.

    Let $\delta_0$ be the current $\delta$. We consider a moment after that where $\delta > \frac{3}{4} \delta_0$. Note that $\delta \leq \delta_0$ by the observation above. After processing line 6 for this loop, $|V_1|$ and $|V_i|$ both become at most
    \begin{equation*}
        |V_1| - \frac{1}{4} \delta = \frac{3}{4} \delta + \left(\frac{n+\alpha+1}{k-2}+1\right) \leq \frac{3}{4} \delta_0 + \left(\frac{n+\alpha+1}{k-2}+1\right)
    \end{equation*}
    by definition of $\delta$. Therefore, the number of $i$'s that $|V_i| > \frac{3}{4} \delta_0 + \left(\frac{n+\alpha+1}{k-2}+1\right)$ will decrease and eventually becomes zero after at most $k$ loops. This means that, after $k$ loops, $\delta$ will be at most $\frac{3}{4}$ times its original value.

    Since $\delta < n$ at the beginning, $\delta$ becomes less than $\frac{1}{k}$ after $k \cdot \lceil \log_{4/3} (nk) \rceil$ loops. However, while $|V_1| > \left\lceil \frac{n+\alpha+2}{k-2} \right\rceil$, we have $\delta \geq \frac{1}{k}$. Therefore, at most $k \cdot \lceil \log_{4/3} (nk) \rceil$ loops will occur.
\end{proof}

Finally, we show that we can compute a $k$-coloring of $G$ such that each color is used in at most $\left\lceil \frac{n+\alpha+2}{k-2} \right\rceil$ vertices, in $O(kn \log n)$ time, which proves \autoref{thm:main}.

\begin{lemma} \label{lem:max-edges}
    $|E(G)| \leq 4n+2\alpha$.
\end{lemma}

\begin{proof}
    If we choose $S \subseteq V$ uniformly at random (out of $2^n$ possible choices), the edges between $S$ and $V \setminus S$ form a bipartite subgraph. By assumption, the number of edges in this subgraph is at most $2n + \alpha$. However, since each edge is chosen with probability $\frac{1}{2}$, the expected number of edges is $\frac{1}{2} |E(G)|$. Therefore, $\frac{1}{2} |E(G)| \leq 2n+\alpha$, which means that $|E(G)| \leq 4n+2\alpha$.
\end{proof}

\begin{lemma} \label{lem:loop-time}
    Each loop of \autoref{alg:main} can be executed in $O(n)$ time.
\end{lemma}

\begin{proof}
    We first discuss the required time complexity to find $i$ such that $\sum_{j=1}^{q_i} \max(b_{i,j}-a_{i,j}, -1)$. The naive implementation that computes all components $C_{i,j}$ with depth-first search (DFS) takes $O(|V(H_i)| + |E(H_i)|)$ time; summing over all $i = 2, \dots, k$ could yield a total time of $\Theta(kn)$ if $|V_1|$ is large. However, it is enough to enumerate all components except isolated vertices; then, at most $2 |E(H_i)|$ vertices are concerned, so the time complexity for DFS becomes only $O(|E(H_i)|)$. For the isolated vertices, we only need to count the unvisited vertices in $V_1$ and $V_i$ and add their contribution to the result $\sum_{j=1}^{q_i} \max(b_{i,j}-a_{i,j}, -1)$; this only takes $O(1)$ time. Therefore, the total time to process line 4 in \autoref{alg:main} is $O(\sum_{i=2}^k |E(H_i)|) = O(|E(H)|) = O(n)$, due to $|E(H)| \leq 2n + \alpha$.

    For the remaining procedure, it takes $O(|V(G)|)$ time to rearrange colors, $O(|E(G)|)$ time to find all edges from vertices of color $1$, and $O(|V(H_i)| + |E(H_i)|)$ time for lines 5 and 6 in \autoref{alg:main}. Here, since $|E(G)| \leq 4n+2\alpha$ (\autoref{lem:max-edges}) and $|E(H_i)| \leq 2n+\alpha$, it takes $O(n)$ time in total.
\end{proof}

\begin{proof}[Proof of \autoref{thm:main}]
    In \autoref{alg:main}, each loop is executed in $O(n)$ time (\autoref{lem:loop-time}) and only $O(k \log n)$ loops will occur (\autoref{lem:num-loops}; note that we can assume $k \leq n$ since the $k > n$ case is trivial). Therefore, \autoref{alg:main} terminates in $O(kn \log n)$ time. Also, due to the loop termination condition, \autoref{alg:main} correctly computes a coloring $\varphi$ such that each color is used on at most $\left\lceil \frac{n+\alpha+2}{k-2} \right\rceil$ vertices.
\end{proof}

\section{Consequences} \label{sec:consequences}

In this section, we discuss the consequences of \autoref{thm:main} for planar graphs and other surfaces.

\subsection{Planar Graphs}

Every bipartite planar graph with $n \geq 3$ vertices has at most $2n-4$ edges (\autoref{prop:bipartite-max}). Hence, for a planar graph $G$, all its bipartite subgraphs $H$ with $|V(H)| \geq 3$ satisfy $|E(H)| \leq 2 |V(H)| - 4$. Therefore, we can use \autoref{thm:main} with $\alpha = -4$. Moreover, by the four-color theorem \cite{AH77}, the planar graph $G$ is guaranteed to be $k$-colorable for $k \geq 4$, and the recent result by \cite{IKM+26} shows that a $4$-coloring of $G$ can be computed in $O(n \log n)$ time. Therefore, the following result holds:

\begin{corollary}
    For $k \geq 4$, every planar graph with $n \geq 3$ vertices admits a $k$-coloring such that each color is used on at most $\left\lceil \frac{n-2}{k-2} \right\rceil$ vertices, and such a coloring can be computed in $O(kn \log n)$ time.
\end{corollary}

This bound exactly matches the lower bound of $\left\lceil \frac{n-2}{k-2} \right\rceil$, achieved with a planar graph $K_{1,1,n-2}$ (\autoref{fig:lowerbound}). In particular, for $k = 4$, every planar graph with $n \geq 3$ vertices admits a 4-coloring such that each color is used on fewer than half of the vertices, and such a coloring can be computed in $O(n \log n)$ time. This is itself stronger than the four-color theorem.

\subsection{Results for Other Surfaces}

The result can also be applied to general surfaces, e.g., the torus and the projective plane. For a surface, the Euler characteristic $\chi$ is defined, where every connected graph $G = (V, E)$ embedded on this surface satisfies $|V| - |E| + |F| \geq \chi$, with $F$ denoting the set of faces. For example, $\chi = 2$ for the plane, $\chi = 1$ for the projective plane, and $\chi = 0$ for the torus and the Klein bottle.

\begin{proposition}[\cite{MT01}] \label{prop:bipartite-max-surface}
    Every bipartite graph with $n \geq 3$ vertices on a connected surface with Euler characteristic $\chi$ has at most $2n-2\chi$ edges.
\end{proposition}

By \autoref{prop:bipartite-max-surface}, we can use \autoref{thm:main} with $\alpha = -2\chi$; since $\chi \leq 2$ holds for every connected surface, the assumption $\alpha \geq -4$ is automatically satisfied. Then, we obtain the following result. In particular, for a closed orientable surface, its Euler characteristic is expressed as $\chi = 2 - 2g$ where $g$ is its genus, meaning the threshold in the corollary below becomes $\left\lceil \frac{n+4g-2}{k-2} \right\rceil$ in this case.

\begin{corollary}
    For $k \geq 3$, given a graph with $n \geq 3$ vertices on a fixed connected surface with Euler characteristic $\chi$, together with its $k$-coloring $\varphi: V \to \{1, \dots, k\}$, one can obtain a $k$-coloring such that each color is used on at most $\left\lceil \frac{n-2\chi+2}{k-2} \right\rceil$ vertices, in $O(kn \log n)$ time.
\end{corollary}

Finally, we note that the result above holds beyond the surface's colorability. For example, it is proved that every graph on the torus is $7$-colorable \cite{RY68}, but if $G$ is $4$-colorable, then there exists a $4$-coloring such that each color is used on at most $\left\lceil \frac{n+2}{2} \right\rceil$ vertices.

\section{Conclusion} \label{sec:conclusion}

In this paper, we showed that every planar graph with three or more vertices admits a $4$-coloring in which no color is used on half or more of the vertices. This result is shown in the form of \autoref{thm:main}, which allows us to generalize to five or more colors and other surfaces. We consider that our result is particularly meaningful when the number of colors $k$ is small. This is because, if $k$ is large enough, the following reduction works:

\begin{enumerate}
    \item If a given $\varphi$ is a coloring where each color is used on at most $\left\lceil \frac{n+\alpha+2}{k-2} \right\rceil$ vertices, just output $\varphi$.
    \item Otherwise, one can obtain an independent set $S \subseteq V$ of size exactly $\left\lceil \frac{n+\alpha+2}{k-2} \right\rceil$ from the largest color class in $\varphi$, and then it remains to find a $(k-1)$-coloring of $G[V \setminus S]$ where each color is used on at most $\left\lceil \frac{n+\alpha+2}{k-2} \right\rceil$ vertices.
\end{enumerate}

This reduction applies only if $G[V \setminus S]$ is $(k-1)$-colorable and its initial $(k-1)$-coloring can be computed efficiently. This is true when $k \geq 5$ for planar graphs, due to an $O(n \log n)$-time algorithm to find a $4$-coloring of planar graphs \cite{IKM+26}. Therefore, $k = 4$ is the essentially important case for planar graphs. Likewise, for general graphs that meet the condition of \autoref{thm:main}, they are \emph{typically} $9$-degenerate by \autoref{lem:max-edges}, and such graphs are $10$-colorable in polynomial time. For these reasons, the essential cases are when $k$ is small, e.g., $k \leq 10$.

Finally, we believe that this paper has some implications beyond this result. First, while the four-color theorem is notoriously complex, one can obtain a meaningful generalization by treating it as a ``black box.'' Second, while Kempe chains have often been used for induction (for the four-color theorem) or for reconfiguration problems, they can also be used as an optimization algorithm in graph coloring. We conclude this paper by posing the following conjectures:

\begin{conjecture}
    Given a planar graph with $n \geq 3$ vertices, one can compute a 4-coloring in which each color is used on fewer than $n/2$ vertices in $O(n)$ time.
\end{conjecture}

\begin{conjecture} \label{conj:grotzsch}
    Every triangle-free planar graph with $n \geq 2$ vertices admits a 3-coloring in which each color is used on at most $n/2$ vertices (a possible extension of Gr\"{o}tzsch's theorem \cite{Gro59}).
\end{conjecture}

Particularly for \autoref{conj:grotzsch}, the example in \autoref{fig:grotzsch} that uses the Gr\"{o}tzsch graph (as a $4$-critical graph) is an important graph that every $3$-coloring has a color that is used on at least $n-10$ vertices. This graph can be embedded in a projective plane, which implies that the conjecture would likely not hold for surfaces other than the plane.

\section*{Acknowledgments}

The authors thank Yuta Inoue, Atsuyuki Miyashita, and Tonan Kamata for helpful discussions, and Akihito Yoneyama for helpful suggestions to improve this paper.

\printbibliography

\end{document}